\documentclass[sigconf]{acmart}
\usepackage[normalem]{ulem}
\usepackage{microtype}
\usepackage{breqn}
\usepackage{mathptmx}
\usepackage{subfig}
\usepackage{textcomp}
\usepackage{multirow}
\usepackage{framed}
\usepackage{amsmath}
\usepackage[ruled,vlined,linesnumbered]{algorithm2e}
\usepackage[beginComment=//~,beginLComment=//~,endLComment=~,italicComments=true,noEnd=true]{algpseudocodex}
\usepackage{listings}
\usepackage{graphicx}
\usepackage{booktabs}
\usepackage[normalem]{ulem}
\usepackage{tikz}
\usepackage{xcolor}
\usepackage{amsthm}
\usepackage{acronym}
\usepackage[capitalize]{cleveref}
\usepackage{enumitem}
\usepackage{printlen}
\usepackage[squaren]{SIunits}
\usepackage{tabularx}
\usepackage{tikz}
\usepackage[normalem]{ulem}
\usepackage{circledsteps}
\usepackage[dvipsnames,svgnames,table]{xcolor}
\usepackage[utf8]{inputenc}
\usepackage{url}
\newtheorem{lemma}{Lemma}[section]
\newtheorem{theorem}{Theorem}[section]
\newtheorem{remark}[theorem]{Remark}
\newenvironment{explanation}{\textit{Explanation.}\ \ignorespaces}{\qed}

\AtBeginDocument{%
  }

\author{Nathan Ng}
\affiliation{%
  \institution{University of Massachusetts Amherst}
  \city{Amherst}
  \state{MA}
  \country{USA}
  \country{}
}
\email{kwanhong@cs.umass.edu}

\author{David Irwin}
\affiliation{%
  \institution{University of Massachusetts Amherst}
  \city{Amherst}
  \state{MA}
  \country{USA}
  \country{}
}
\email{irwin@ecs.umass.edu}

\author{Ananthram Swami}
\affiliation{%
  \institution{DEVCOM Army Research Lab}
  \city{Adelphi}
  \state{Maryland}
  \country{USA}
}
\email{ananthram.swami.civ@army.mil}

\author{Don Towsley}
\affiliation{%
  \institution{University of Massachusetts Amherst}
  \city{Amherst}
  \state{MA}
  \country{USA}
  \country{}
}
\email{towsley@cs.umass.edu}

\author{Prashant Shenoy}
\affiliation{%
  \institution{University of Massachusetts Amherst}
  \city{Amherst}
  \state{MA}
  \country{USA}
  \country{}
}
\email{shenoy@cs.umass.edu}

\setcopyright{acmlicensed}
\renewcommand\footnotetextcopyrightpermission[1]{}

\copyrightyear{2018}
\acmYear{2018}
\acmDOI{XXXXXXX.XXXXXXX}
\acmConference[Conference acronym 'XX]{Make sure to enter the correct
  conference title from your rights confirmation email}{June 03--05,
  2018}{Woodstock, NY}
\acmISBN{978-1-4503-XXXX-X/2018/06}

\begin{document}
\title[To Offload or Not To Offload: Model-driven Comparison of Edge-native and On-device Processing]{To Offload or Not To Offload: Model-driven Comparison of Edge-native and On-device Processing In the Era of Accelerators}

\begin{abstract}
Computational offloading is a promising approach for overcoming resource constraints on client devices by moving some or all of an application's computations to remote servers. With the advent of specialized hardware accelerators, client devices can now perform fast local processing of specific tasks, such as machine learning inference, reducing the need for offloading computations. However, edge servers with accelerators also offer faster processing for offloaded tasks than was previously possible. 
In this paper, we present an analytic and experimental comparison of on-device processing and edge offloading for a range of accelerator, network, multi-tenant, and application workload scenarios, with the goal of understanding when to use local on-device processing and when to offload computations.
We present models that leverage analytical queuing results to derive explainable closed-form equations for the expected end-to-end latencies of both strategies, which yield precise, quantitative performance crossover predictions that guide adaptive offloading.
We experimentally validate our models across a range of scenarios and show that they achieve a mean absolute percentage error of 2.2\% compared to observed latencies.   We further use our models to develop a resource manager for adaptive offloading and show its effectiveness under variable network conditions and dynamic multi-tenant edge settings.

\end{abstract}

\maketitle

\section{Introduction}
\label{sec:intro}

Over the past decade, cloud computing has become the predominant approach for running online services in domains ranging from finance to entertainment. 
In recent years, a new class of online services has emerged that requires low-latency processing to meet end-user requirements. 
Examples of such applications include autonomous vehicles, interactive augmented and virtual reality (AR/VR), and human-in-the-loop machine learning inference \cite{Mosqueira2022}. 
Edge computing, a complementary approach to cloud computing, is a promising method for meeting the needs of such latency-sensitive applications. 
Conventional wisdom has held that since edge resources are deployed 
at the edge of the network and are 
closer to end-users, edge-native processing can provide lower end-to-end latencies than cloud processing. 
However, recent research \cite{Ali2021,WangB2024} has shown that this conventional wisdom does not always hold---in certain scenarios, edge processing can yield \emph{worse} end-to-end latency despite its network latency advantage over the cloud. 
In particular, since edge clusters are resource-constrained relative to the cloud, research has shown that resource bottlenecks can arise under time-varying workloads and increase the processing latency at the edge, negating its network latency advantage and requiring careful resource management to maintain its benefits \cite{WangB2024}. 

A related trend is the emergence of {\em computational offloading} from local devices to remote resources such as those at the edge.\footnote{Computational offloading to the cloud is also common, but we focus on edge offloading in this paper.} 
Computational offloading \cite{Satya2009,Davis2004} involves using nearby edge resources to take on some, or all, of an application's processing demands.  Such offloading is performed to overcome resource constraints on a device by leveraging the more abundant processing capacity at edge servers. 
Computational offloading is well-established in mobile computing where mobile devices, which are resource- or battery-constrained, offload local computations to edge servers \cite{Ren2019}. 
In this case, local devices pay a so-called mobility penalty \cite{Satyanarayanan2019}, which is the additional network latency to access edge servers, but then benefit from faster edge processing that outweighs this latency overhead. 
Conventional wisdom holds that remote edge processing is faster than local processing despite the network overhead of accessing remote resources due to significant resource constraints at local devices. 

With the advent of hardware accelerators in recent years, the conventional wisdom about the benefits of computational edge offloading for local devices needs to be rethought. 
For example, the rise of the so-called AI PC with onboard AI accelerators~\cite{intel_ai_pc} has enabled local AI processing for many applications without the need for edge offloading. 
Similarly, mobile phones are equipped with neural accelerators \cite{apple_neural_engine} that enable efficient local inference, while IoT devices such as cameras have specialized accelerators to perform visual tasks (e.g., facial recognition) locally \cite{netatmo_camera}. 
At the same time, edge servers can also be equipped with more powerful accelerators, such as GPUs,  enabling offloaded tasks to also run more efficiently or faster than before.
Thus, in the era of programmable accelerators, the question of whether to offload computations from local devices to the edge and when it is advantageous to do so needs to be reexamined.
In particular, such a reexamination needs to address several questions. 
(i) Under what scenarios will local on-device processing using accelerators outperform edge offloading?
(ii) Are there still scenarios where edge offloading has an advantage over local processing?
(iii) How does multi-tenancy at edge servers affect the decision to offload or execute on-device?

Motivated by the above questions, this paper takes an analytic model-driven approach to compare local and edge processing in the era of accelerators.  
Our hypothesis is that there is no a priori clear advantage of edge processing over on-device processing in the presence of accelerators, and that the choice of where to perform the processing depends on hardware, network, and workload characteristics.  
The novelty of this work lies in deriving closed-form equations for the expected end-to-end latency of each strategy, which are explainable and require no workload-specific training. These equations yield precise, quantitative performance-crossover predictions that can be embedded in OS schedulers and resource managers to dynamically decide whether to run an application component on the device or at the edge. 
In developing our analytic models and validating our hypothesis, this paper makes the following contributions.
\begin{itemize}[left=0pt]
\item We develop queuing theory–based models to model the behavior of device and edge accelerators when processing requests in on-device and edge offloading scenarios.
Our modeling adopts a two-level approach: 
it first uses estimated service times obtained through profiling or prediction techniques as input, and then embeds them into queuing models to predict end-to-end response times.
We develop closed-form analytic bounds using these models to determine when one approach outperforms the other under different hardware, network, and workload scenarios. 
We further extend the analysis to model multi-tenant edge servers serving multiple clients, and demonstrate how the models can be extended to capture device–edge collaborative (i.e., split) processing.

\item  We experimentally validate our models and bounds across diverse device and edge accelerators, network configurations, and workloads spanning multiple model architectures, including deep neural network (DNN), recurrent neural network (RNN), and large language model (LLM).
The results show that our models accurately predict on-device and offloading performance, achieving a 2.2\% mean absolute percentage error, with 91.5\% of predictions within $\pm$5\% and all within $\pm$10\% of the observed latency.

\item We develop a resource manager that leverages our models to adapt between on-device processing and edge offloading in response to real-world dynamics. We present two case studies demonstrating that our models enable effective adaptation to changing network conditions and fluctuating request rates at multi-tenant edge servers.
\end{itemize}

\section{Background}\label{sec:background}
This section provides background on computational offloading, hardware accelerators, and on-device processing.

\subsection{Computational Offloading}
Computational offloading is a well-known approach where client devices that are resource-constrained offload some or all of their application workload to a remote server. 
In the case of mobile devices such as AR/VR headsets, tablets, or smartphones, offloading can use nearby edge resources to perform latency-sensitive tasks \cite{Satya2009,Davis2004}. 
While such edge offloading involves a network hop to the edge server (also known as the mobility penalty \cite{Satyanarayanan2019}), subsequent processing at the edge is assumed to be much faster since edge resources are much less constrained than those of devices. 
In general, computational offloading offers benefits when local processing is constrained and the benefits of accessing faster processing at the edge or in the cloud outweigh the cost of additional network latency. 

\subsection{Hardware Accelerators}
In recent years, a wide variety of accelerators have been developed to speed up diverse compute tasks, particularly machine learning workloads. 
Table 1 depicts the characteristics of several commonly deployed device and edge accelerators. 
For example, NVIDIA TX2, Orin Nano, and A2 are all tailored for AI workloads such as deep learning and computer vision. 
With their small form factor and energy-efficient design, accelerators are increasingly integrated into mobile devices. 
For example, Apple devices include a Neural accelerator to accelerate machine learning tasks locally on the device \cite{apple_neural_engine}.  
With accelerators, fast on-device processing of tasks such as machine learning inference has become feasible without resorting to edge offloading. 
Meanwhile, edge servers can also be equipped with even more powerful GPUs.  
This enables offloaded tasks to also be accelerated at the edge, making edge offloading useful for more compute-intensive tasks that a device may not be able to handle using local resources.

\begin{table}[t]
\caption{Characteristics of common hardware accelerators.}
\vspace{-.5em}
\centering
\begin{small}
\begin{tabular}{lccccc}
\hline
\textbf{Accelerator} & \textbf{Power} & \textbf{Memory} & \textbf{Workload} \\ \hline
Google Edge TPU \cite{edgetpu2019} & 2 W & 8 MB & ML Inference \\ 
NVIDIA Jetson TX2 \cite{tx22017}& 15 W & 8 GB & GPU compute \\ 
NVIDIA Jetson Orin Nano \cite{orin2023}& 15 W & 8 GB & GPU compute \\ 
NVIDIA A2 GPU \cite{a22021}& 60 W & 16 GB & GPU compute \\ \hline
\end{tabular}
\end{small}
\label{tab:accelerators}
\end{table}
\subsection{On-device Processing versus Edge Offloading}
\begin{figure*}[t]
\centering
\subfloat[Edge offloading queuing model\label{fig:edge_model}]{
    \includegraphics[width=1.116\columnwidth]{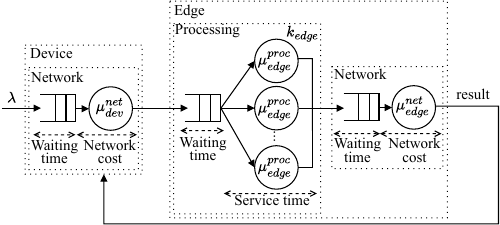}
}
\hspace{3.5em}
\subfloat[On-device processing queuing model\label{fig:device_model}]{
    \includegraphics[width=0.595\columnwidth]{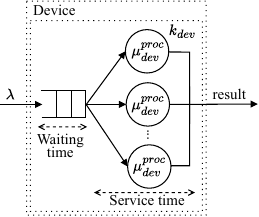}
}
    \caption{Modeling request execution using (a) edge offloading and (b) on-device processing.}
    \label{fig:queueing_models}
\end{figure*}
As noted above, the availability of accelerators on client devices has enabled on-device processing of many tasks without the need to offload those tasks to the edge.
Since local processing does not incur any network hops, it will often be faster than remote processing which incurs additional network latency.  
However, since edge servers can also be equipped with GPUs and other accelerators, in some compute-intensive scenarios, edge processing may still provide benefits over local processing on device accelerators.
Thus, in the era of accelerators, local on-device processing may be faster in some scenarios, while offloading to edge accelerators may still be faster in others (after accounting for network latency).

Even with extensive offline profiling, selecting the optimal execution strategy remains challenging due to real-world dynamics. 
On the one hand, mobile devices are subject to network variability such as fluctuations in wireless signal strength.
Since edge offloading involves transmitting input data, network fluctuations can introduce unpredictability in network delays, thereby affecting offloading performance. 
Additionally, devices running applications with prolonged operating durations such as live object detection may face battery constraints.
In low-power mode, processors may be underclocked by the OS to conserve energy, extending on-device execution times. 
On the other hand, edge servers are {\em multi-tenant} systems shared by multiple devices, which can lead to interference among clients and higher latencies. 
Since devices may independently decide when and whether to offload tasks, the edge server’s load can fluctuate dynamically with the number and type of incoming requests.
Unlike traditional cloud data centers, edge servers are often more resource-constrained and cannot elastically scale all applications under high load. As a result, workload spikes in multi-tenant scenarios can increase queuing delays and significantly extend end-to-end execution times. 
Moreover, edge servers are often multiplexed across multiple applications to maximize resource utilization.
GPUs and other accelerators designed for edge servers often lack isolation mechanisms such as GPU virtualization that are available in cloud-based GPUs \cite{MIG_nvidia}, leading to performance interference between applications.
Consequently, execution times at the edge can vary, making it challenging to decide whether to offload under dynamic workloads.

The above observations motivate the primary research question addressed in this paper: \textit{Considering the hardware processing capabilities of the device and edge server, along with the workload characteristics and real-world dynamics, under what scenarios does one strategy outperform the other?}  Specifically, does edge offloading still have a role to play when devices are equipped with on-board accelerators? How do workload dynamics, network dynamics, and edge multi-tenancy impact each strategy and their relative performance with respect to one another? 
To address these questions, the next section develops analytic models for the end-to-end execution time of edge offloading and on-device processing, capturing the key factors that impact application performance.

\section{Model-driven Performance Comparison}
\label{sec:models}
In this section, we develop analytic queuing models for edge offloading and on-device processing, and derive closed-form expressions for their expected latencies to compare the two approaches and identify when one outperforms the other.

\subsection{Analytic Queuing Models}
\label{sec:analysis}

Our analytic models are derived from queuing theory, a well-established mathematical tool for modeling response times experienced by requests in computing systems. Queuing theory has previously been used to model the performance of edge and cloud applications \cite{Ali2021,Harchol2013,Urgaonkar2005}, and here will serve as the foundation for analyzing the end-to-end latency of requests---from the time a task enters the system until it completes---under both edge offloading and on-device processing.

\noindent\textbf{Edge offloading.} 
To model the end-to-end latency of a request under edge offloading, we need to model the times spent by the request both at the device and at the edge server, which we model as two separate queuing systems as shown in Figure \ref{fig:edge_model}. We assume that a request arrives at the device and is then forwarded (``offloaded'') to the edge over the network, incurring network queuing and transmission delays. 
Upon arrival at the edge, the request is queued for processing and is subsequently scheduled for execution. 
Since accelerators can process multiple requests in parallel, we model the edge accelerator as a parallel system with a degree of parallelism $k_{edge}$.
After processing, the result is sent back to the device, which incurs network queuing and transmission delays on the reverse path.


\noindent\textbf{On-device processing.} Similar to edge offloading, we model on-device request processing as a queuing system shown in Figure \ref{fig:device_model}. 
Incoming requests generated by the device first enter a local device queue and are then scheduled onto one of the accelerator cores, where they are processed with a degree of parallelism $k_{dev}$. 
When processing completes, the result is returned to the local application, and the request exits the system.
This model captures an application where a client device generates tasks, uses its accelerator cores for local execution, and sends the results back to the application.


In queuing theory, the processing time of a request at each queue that it encounters consists of two key components: queuing delay (i.e., the time spent waiting in the queue) and service time (i.e., the time taken to execute the request on the accelerator).
Let $w_{dev}^{proc}$, $w_{edge}^{proc}$, $s_{dev}$, and $s_{edge}$ represent the wait times and request service times at the device and edge, respectively. 
Edge offloading additionally incurs network latency as requests are sent to the edge and results are returned to the device. 
Let $w_{dev}^{net}$ and $w_{edge}^{net}$ denote the respective network queuing delays on the device and the edge.
Also, let $n_{req}$ denote the network transmission delay of sending a request from the device to the edge, and $n_{res}$ denote the transmission delay of the response along the reverse path.

With the above notations, we now model the end-to-end latency of requests executed using each strategy: 
\begin{equation}
T_{edge} = w_{dev}^{net} + n_{req} + w_{edge}^{proc} + s_{edge} + w_{edge}^{net} + n_{res} \label{eq:edge}
\end{equation}
and 
\begin{equation}
T_{dev} = w_{dev}^{proc} + s_{dev} \label{eq:device}
\end{equation}
where $T_{edge}$ and $T_{dev}$  denote the end-to-end request latency under edge offloading and on-device processing, respectively. 

\subsection{Deriving Service Times on Accelerators\label{sec:service_time}}
Our models require the service time of a workload as input to compute end-to-end response times. 
The service time reflects the computational demand on the target accelerator and represents the time needed to process a request on that platform. 
It can be obtained either through empirical profiling or via a predictive deep learning model.
For profiling, many applications already provide function-level performance logging or profiling hooks \cite{nvidia2023fastertransformer}, which can be used to derive service times by analyzing execution metrics. 
Alternatively, the workload can be benchmarked directly on the target hardware using tools such as \cite{Reddi2020,nvidiasmi}.
For model-based prediction, a neural network can be trained to predict the service time of a model on a given hardware. 
For example, prior work predicts DNN inference latency by analyzing network structure and parameter size \cite{Wang2024,Yao2018,Kang2017}, RNN inference latency by examining temporal dependencies and sequence unrolling behavior \cite{Ji2024,Li23}, and LLM inference latency by using internal layer embeddings to predict remaining output length \cite{Rana2024}.
While any of these prior techniques can be used with our analytic models, in our experiments we adopt a simple neural network from \cite{Kang2017} to predict service times.

\subsection{Deriving and Analyzing End-to-End Latency\label{sec:deriving}}
In this section we compute the expected latency for edge offloading and on-device processing.
Our analysis focuses on applications that require results on the device, such as AR/VR workloads. 
 While the models assume edge offloading returns results to the device which incurs network delay, they can be generalized to applications where results are consumed on the edge (e.g., IoT sensing) by omitting this network delay. 
In the following, we leverage classical queuing theory to fundamentally understand the conditions under which one strategy outperforms the other.

\subsubsection{Comparing Edge Offload and On-Device Execution\label{sec:dedicated}}
We begin our analysis by examining the conditions under which edge offloading results in higher average latency compared to on-device processing.
For workloads that benefit from hardware acceleration such as those utilizing GPUs, processing times can be significantly reduced compared to processing on general-purpose CPUs. 
In the following discussion, we refer to such workloads as accelerator-driven workloads. 
While our analysis primarily focuses on DNN inference, a representative component commonly found in AR/VR applications, the results generalize to other AI workloads as discussed in Sec \ref{sec:discussions}.
Table \ref{tab:notations} summarizes the notation used in our models.
\begin{table}[!t]
\centering
\caption{Notations used in analytic models.}
\vspace{-.5em}
\begin{tabular}{ll}
\hline
Notation & Description \\ \hline  
$T_{dev}$ & End-to-end latency of on-device processing  \vspace{0.19em} \\
$T_{edge}$ & End-to-end latency of edge offloading  \vspace{0.19em} \\
$\lambda$ & Request arrival rate \vspace{0.19em} \\ 
$B$ & Bandwidth of the device-to-edge network path \vspace{0.19em} \\ 
$D_{req}, D_{res}$ & Payload sizes of request/result  \vspace{0.25em} \\  
$w_{dev}^{net},w_{edge}^{net}$ & Network wait time at the device/edge \vspace{0.13em} \\
$n_{req},n_{res}$ & Network transmission time of request/response  \vspace{0.25em} \\
$\mu_{dev}^{net}, \mu_{edge}^{net}$ & Service rate of device/edge NIC  \vspace{0.13em} \\ 
$k_{dev},k_{edge}$ & Parallelism level of processors at device/edge \vspace{0.19em} \\ 
$w_{dev}^{proc},w_{edge}^{proc}$ & Processing wait time at the device/edge  \vspace{0.13em} \\
$s_{dev},s_{edge}$ & Service time at the device/edge  \vspace{0.25em} \\ 
$\mu^{proc}_{dev}, \mu^{proc}_{edge}$ & Service rate of device/edge  \vspace{0.13em} \\ 
\hline
\end{tabular}
\label{tab:notations}
\end{table}
\begin{lemma}
\label{eq:d_vs_e_lemma}
For accelerator-driven workloads, edge offloading incurs a higher average end-to-end latency than on-device processing when 
\begin{dmath}
      s_{dev} - s_{edge} 
      <
      \frac{\lambda}{\mu_{dev}^{net} (\mu_{dev}^{net} - \lambda)}
      + \frac{\lambda}{\mu_{edge}^{net} (\mu_{edge}^{net} - \lambda)}
      + \frac{D_{req} + D_{res}}{B}  
      + \frac{1}{2} \left( \frac{1}{k_{edge}\mu_{edge}^{proc} - \lambda} - \frac{1}{k_{edge}\mu_{edge}^{proc}} \right) 
      - \frac{1}{2} \left( \frac{1}{k_{dev}\mu_{dev}^{proc} - \lambda} - \frac{1}{k_{dev}\mu_{dev}^{proc}} \right) 
      \label{eq:lemma1}
\end{dmath}
\end{lemma}

\begin{proof}
Edge offloading will offer a worse latency than on-device processing when $T_{device} < T_{edge}$. That is 
\begin{equation}
     w_{dev}^{proc} + s_{dev} < w_{dev}^{net} + n_{req} + w_{edge}^{proc} + s_{edge} + w_{edge}^{net} + n_{res} 
\end{equation}

Let $D_{req}$ and $D_{res}$ denote the request and result payload size, respectively.
Further, let $B$ denote the bandwidth of the device-to-edge network path.
The above inequality reduces to 
\begin{equation}
     s_{dev} - s_{edge} < w_{dev}^{net} + w_{edge}^{net} + \frac{D_{req}}{B} + \frac{D_{res}}{B} + w_{edge}^{proc} - w_{dev}^{proc}
     \label{eq:d_vs_e_1}
\end{equation}

As shown in \cite{Liang2023}, DNN inference on accelerators exhibits deterministic service times (denoted as $D$), as the number of operations per request remains constant. 
Since typical accelerators have a small degree of parallelism $k$ \cite{Amert2017}, the queuing system can be modeled as an $M/D/k/FCFS$ system, where request arrivals follow a Poisson process. 
While no closed-form solution exists for the expected wait time in such a system, a common approximation is to aggregate the service rate across the degree of parallelism $k$ \cite{Shah2018,Shah2020}.
This reduces the model to an $M/D/1/FCFS$ system with aggregated service rate $k\mu$.
Its expected wait time can be derived from the Pollaczek-Khinchin (P-K) formula \cite{Gross2008,Harchol2013} and is given by 
\begin{equation}
        E[w_{M/D/1}] =\frac{1}{2} \left( \frac{1}{k\mu - \lambda} - \frac{1}{k\mu} \right) \label{eq:md1}
\end{equation}
On the other hand, network interfaces with a single controller can be modeled as $M/M/1/FCFS$ systems with expected wait times given by 
\begin{equation}
        E[w_{M/M/1}] = \frac{1}{\mu - \lambda} - \frac{1}{\mu} \label{eq:mm1}
\end{equation}
Let $k_{dev}$ and $k_{edge}$ denote the degree of parallelism seen at the accelerator on the device and the edge, respectively. 
Applying these two results to (\ref{eq:d_vs_e_1}), we get 
\begin{dmath}
      s_{dev} - s_{edge} < \frac{1}{\mu_{dev}^{net} - \lambda_{dev}^{net}} - \frac{1}{\mu_{dev}^{net}}
      + \frac{1}{\mu_{edge}^{net} - \lambda_{edge}^{net}} - \frac{1}{\mu_{edge}^{net}}
      + \frac{D_{req} + D_{res}}{B}  
      + \frac{1}{2} \left( \frac{1}{k_{edge}\mu_{edge}^{proc} - \lambda_{edge}} - \frac{1}{k_{edge}\mu_{edge}^{proc}} \right) 
      - \frac{1}{2} \left( \frac{1}{k_{dev}\mu_{dev}^{proc} - \lambda_{dev}} - \frac{1}{k_{dev}\mu_{dev}^{proc}} \right) 
      \label{eq:d_vs_e_2}
\end{dmath}

Offloading is required only when a request arrives, which gives us $\lambda^{net}_{dev} = \lambda$. 
For a queuing system to be stable, throughput (i.e., task completion rate) must equal the task arrival rate $\lambda$ under steady-state conditions. 
When tasks are offloaded to an edge server, results are transmitted to the device only upon task completion, which yields $\lambda_{edge}^{net} = \lambda$.  
Moreover, regardless of the execution strategy, the application sees the same arrival rate, i.e., \( \lambda_{dev}^{proc} = \lambda_{edge}^{proc} = \lambda \). 
Substituting the above into (\ref{eq:d_vs_e_2}), we obtain
\begin{dmath}
      s_{dev} - s_{edge} < \frac{\lambda}{\mu_{dev}^{net} (\mu_{dev}^{net} - \lambda)}
      + \frac{\lambda}{\mu_{edge}^{net} (\mu_{edge}^{net} - \lambda)}
      + \frac{D_{req} + D_{res}}{B}  
      + \frac{1}{2} \left( \frac{1}{k_{edge}\mu_{edge}^{proc} - \lambda} - \frac{1}{k_{edge}\mu_{edge}^{proc}} \right) 
      - \frac{1}{2} \left( \frac{1}{k_{dev}\mu_{dev}^{proc} - \lambda} - \frac{1}{k_{dev}\mu_{dev}^{proc}} \right) 
\end{dmath}
which completes the proof. 
\end{proof}
The above result implies the following two remarks.

\begin{remark}
\label{eq:remark1}
On-device processing is likely to outperform edge offloading for workloads with lower computational demand. 
\end{remark}
\begin{explanation}
Service time, the time spent processing a request on an accelerator, is determined by the workload's computational demand divided by the hardware's processing capacity.  For any given hardware configuration, requests with lower computational demands have to be processed much faster at the edge to overcome the network penalty for edge offloading to outperform on-device processing. 
For requests with lower demand, the difference in average service times $s_{dev} - s_{edge}$ shrinks proportionally as the disparity in processing capacity between the device and the edge becomes less impactful.  Further, since the service rate is the inverse of service time in queuing theory (i.e., $s=1/\mu$), the terms with $\mu_{dev}^{proc}$ and $\mu_{edge}^{proc}$ on the right-hand side of the inequality in (\ref{eq:lemma1}) will also become smaller since $\mu$ increases for lighter requests.  Meanwhile, the terms representing the network overheads remain unchanged, as they depend on the payload size rather than the request's processing demand. As a result, the inequality in (\ref{eq:lemma1}) is more likely to hold for requests with lighter processing demand.
\end{explanation}

\begin{remark}
\label{eq:remark2}
On-device processing is likely to outperform edge offloading on slower networks or when workloads have larger request or result payloads. 
\end{remark}
\begin{explanation}
The term \( \frac{D_{req} + D_{res}}{B} \) in (\ref{eq:lemma1}) increases when either the network bandwidth \( B \) decreases or the request and result sizes \( D_{req} \) and \( D_{res} \) become larger, thereby increasing the right-hand side of the inequality.  
In addition, longer transmission times reduce the NIC service rates at the device and edge ($\mu_{dev}^{net}$ and $\mu_{edge}^{net}$), further increasing the right-hand side.  
As a result, the inequality is more likely to hold under slower networks or for tasks with larger payloads, making on-device processing faster than edge offloading.
\end{explanation}

\noindent\textbf{Practical takeaways. } 
Lemma \ref{eq:d_vs_e_lemma} establishes a  bound on the average service time difference between the device and the edge for edge offloading to be more effective.  The bound and the resulting inequality depend on several factors: (i) relative processing speeds of the device and the edge $s_{dev}$ and $s_{edge}$ (and the corresponding service rates $\mu_{dev}$ and $\mu_{edge}$), (ii) overall workload denoted by request rate $\lambda$, (iii) network transmission overheads.  In general, the edge processing capacity must be larger than the device processing capacity by a factor that overcomes the network penalty of offloading; otherwise, on-device processing becomes preferable to offloading.

Moreover, the request rate $\lambda$ impacts offloading decisions in two competing ways. 
First, as $\lambda$ increases, both the edge and the device experience longer queuing delays, though the edge queue grows more slowly due to its faster processing speed, which favors offloading.
Second, under slow networks or for tasks with large payloads, a high $\lambda$ increases network queuing delay, which penalizes offloading.
Consequently, the optimal strategy depends on the workload’s payload size relative to network bandwidth and the performance gap between the device and the edge.

Remark \ref{eq:remark1} indicates that edge offloading is likely to outperform on-device processing for workloads with high computational demand as the edge’s processing speed advantage dominates, while on-device processing is more efficient for lighter workloads since network overhead can outweigh the edge’s advantage.
Remark \ref{eq:remark2} shows that edge offloading becomes less efficient with limited bandwidth or large payloads, which increase network queuing and transmission delays, whereas on-device processing remains unaffected.

\noindent\textbf{Extending to collaborative processing.}
Collaborative (i.e., split) processing involves partially processing requests locally on the device before offloading them to a remote edge server to minimize the amount of data sent over the network \cite{Almeida2022,Banitalebi2021,Hu2019,Huang2020,Kang2017,Li2018,Yang2022,Chen2025}.
To compute the end-to-end response time of such a strategy, we combine our analytic models for on-device processing and edge offloading into a tandem queuing network that accounts for both the partial local computation and the subsequent transmission and processing delays at the edge server.
Specifically, a request first enters the components modeled in Figure \ref{fig:device_model} and is partially processed with a service time $s'_{dev}$ representing the partial local processing time. 
Then, the intermediate output of size $D_{inter}$ passes through the components modeled in Figure \ref{fig:edge_model}, with a network transmission time of $D_{inter}/B$ and a service time $s'_{edge}$ representing the remaining computation time on the edge. 
This combined model can then be used to estimate performance and inform decision-making when selecting between local, edge, or collaborative processing.

\subsection{Impact of Multi-Tenancy}
While the analysis in Lemma \ref{eq:d_vs_e_lemma} implicitly assumes an edge server that is dedicated to a client device,
in practice, edge servers will be multiplexed across multiple client devices and service multiple applications that offload their work. 
Hence, we extend our model-driven comparison to multi-tenancy scenarios.

\noindent\textbf{Modeling edge application multiplexing. }
In the case where each edge server is dedicated to a single client device, the workload (i.e., the request arrival rate) at the edge and at the device $\lambda^{proc}_{dev}$ and $\lambda^{proc}_{edge}$ are identical. 
However, when an edge server is multiplexed across $m$  devices, its workload will be the sum of the individual device workloads under the offloading case, as shown below: 
\begin{center}
\includegraphics[width=.58\columnwidth]{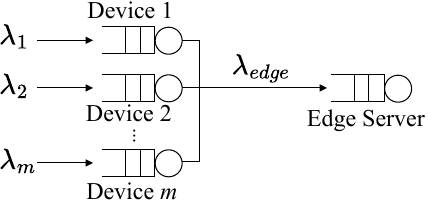}
\end{center}
This higher workload will increase the queuing delays at the edge and may cause the end-to-end offloading latency to exceed the on-device processing latency.
Since the request arrivals at each client device are governed by independent Poisson processes, and the composition of independent Poisson processes results in a new Poisson process \cite{pishro2014}, it follows that the aggregate edge workload is also Poisson with an arrival rate $\lambda_{edge}$ equal to the sum of the workloads of all devices. That is,
$\lambda_{edge} = \sum_{i=0}^m \lambda_i .$, where $\lambda_i$ denotes the workload at the $i-th$ device.

We assume that the multiplexed applications have different service time demands, and the aggregate service time distribution seen at the edge will be arbitrary. 
In this case, we model the edge system as an $M/G/1$ system, where $M$ denotes Poisson arrivals and $G$ denotes the general (arbitrary) service time distribution seen at the edge.
This extension captures the potentially large variability in service times across applications in multi-tenant environments. 
Let $s_i$ denote the average service time of requests of device $i$ when processed at the edge, then $s_{edge}$ is the weighted average of these service times. That is, $s_{edge} = \sum_i \frac{\lambda_i}{\lambda_{edge}}  s_i$ with the effective service rate of the edge given by $\mu_{edge}^{proc} = 1 / s_{edge}$. The aggregate utilization of the edge server is then $\rho_{edge} = \lambda_{edge} / \mu_{edge}^{proc}$. This leads us to the following lemma for the multi-tenant case.


\begin{lemma} For multi-tenant edge servers, the average end-to-end latency of edge offloading is higher than on-device processing when 
\begin{dmath}
      s_{dev} - s_{edge} 
      <
      \frac{\lambda_{dev}}{\mu_{dev}^{net} (\mu_{dev}^{net} - \lambda_{dev})}
      + \frac{\lambda_{edge}}{\mu_{edge}^{net} (\mu_{edge}^{net} - \lambda_{edge})}
      + \frac{D_{req} + D_{res}}{B}  
      +  \frac{\rho_{edge} + \lambda_{edge} k_{edge}\mu_{edge}^{proc} \text{Var}[s_{edge}]}{2(k_{edge}\mu_{edge}^{proc} - \lambda_{edge})}
      - \frac{1}{2} \left( \frac{1}{k_{dev}\mu_{dev}^{proc} - \lambda_{dev}} - \frac{1}{k_{dev}\mu_{dev}^{proc}} \right) 
      \label{eq:lemma2}
\end{dmath}   
\end{lemma}
\begin{proof}
    The expected wait time of $M/G/1/FCFS$ systems is well-known and is given by P-K formula \cite{Gross2008,Harchol2013} 
    \begin{dmath}        
        E[w_{M/G/1}] = \frac{\rho + \lambda \mu \text{Var}[s]}{2(\mu - \lambda)} 
        \label{eq:mg1}
    \end{dmath}
    Substituting (\ref{eq:mg1}) for $w_{edge}^{proc}$, (\ref{eq:md1}) for $w_{dev}^{proc}$, and (\ref{eq:mm1}) for $w_{dev}^{net}$ and $w_{edge}^{net}$ into (\ref{eq:d_vs_e_1}) completes the proof.
\end{proof}

\noindent\textbf{Practical takeaways. }
The expected wait time in $M/G/1/FCFS$ systems is heavily influenced by the variance in service times, meaning that co-located applications with highly variable service demands can significantly increase queuing delays.
It is therefore important for application designers to consider service time variability when assigning applications to edge servers.
Grouping applications with similar processing demands can help reduce overall variability, thereby minimizing queuing delays and improving edge offloading performance.
Further, the above analysis assumes that isolation features such as GPU virtualization, which are available on server-grade  GPUs, \emph{are not available on edge accelerators}. If such features were available, each device could be assigned to an isolated virtual GPU on the edge. In this special case, the offloading situation reduces to the one similar to that considered in Sec \ref{sec:dedicated} where each virtual GPU provides $\frac{1}{m}$-th of the capacity to each device.
However, many edge accelerators lack such features, resulting in a shared accelerator that services a combined workload.

\subsection{Discussion} \label{sec:discussions}
\noindent\textbf{Generalizing to other ML workloads. }
While the above analysis focuses on DNN workloads modeled with deterministic service times, our models can be extended to other accelerator-driven workloads with different service time characteristics by incorporating appropriate queuing formulations. 
For example, recurrent models have variable execution times that depend on input length.
Similarly, LLM inference exhibits variable service times due to variations in the number of autoregressive decoding steps.
To capture such variability, we model the processing stages as $M/M/1$ queues with exponential service times and effective service rates $k\mu$. 
We utilize $M/M/1$ rather than $M/M/k$ because the standard $M/M/k$ derivation relies on birth-death processes and factorials that require $k$ to be a discrete integer.
Based on our empirical observations, accelerator parallelism is more fine-grained and can be more accurately represented by treating $k$ as a continuous multiplier for the service rate.
Under this formulation, we obtain the following lemma.
\begin{lemma}
\label{eq:d_vs_e_lemma}
For recurrent models and LLM workloads, edge offloading incurs a higher average end-to-end latency than on-device processing when 
\begin{dmath}
      s_{dev} - s_{edge} 
      <
      \frac{\lambda}{\mu_{dev}^{net} (\mu_{dev}^{net} - \lambda)}
      + \frac{\lambda}{\mu_{edge}^{net} (\mu_{edge}^{net} - \lambda)}
      + \frac{D_{req} + D_{res}}{B}  
      + \frac{1}{k_{edge}\mu_{edge}^{proc} - \lambda} - \frac{1}{k_{edge}\mu_{edge}^{proc}} 
      - \frac{1}{k_{dev}\mu_{dev}^{proc} - \lambda} + \frac{1}{k_{dev}\mu_{dev}^{proc}}
      \label{eq:lemma1}
\end{dmath}
\end{lemma}
\begin{proof}
Substituting the expected wait time of $M/M/1$ systems (\ref{eq:mm1}) for $w_{dev}^{net}, w_{edge}^{net}, w_{edge}^{proc}, w_{dev}^{proc}$ into (\ref{eq:d_vs_e_1}) completes the proof. 
\end{proof}
\noindent By selecting queueing formulations that match workload variability, the framework accommodates diverse ML workloads and provides accurate response-time estimates, as demonstrated in Sec~\ref{sec:RNN}.


\noindent\textbf{Model limitations. }
Our analysis assumes request arrivals follow a Poisson process and are independent, which may not hold under bursty arrivals. 
To address this, our models can adopt a $G/G/1$ system, which accommodates arbitrary arrival and service time distributions. 
Although no closed-form solution exists for the expected wait time in this setting, Marshall et al. \cite{marshall1968} provide an upper bound 
\begin{dmath}
      E[w_{G/G/1}] \leq \frac{\lambda (\sigma_a^2 + \sigma_s^2)}{2(1-\rho)}
\end{dmath}  
where $\sigma_a^2$ and $\sigma_s^2$ are the variances of the request interarrival time and service time distributions, respectively. 
Additionally, the queuing results used in our analysis assume that each queue has infinite capacity. 
For scenarios with bounded queue sizes, we can incorporate results for finite-buffer queuing systems \cite{Harchol2013} to model the impact of such capacity constraints.

\section{Model Validation}
In this section, we empirically validate the accuracy of our analytic models for predicting the end-to-end response times of edge offloading, on-device processing, and collaborative processing across a wide range of configurations.

\subsection{Experimental Setup}

\noindent{\textbf{Workloads.}} 
\label{sec:workloads}
We validate our model’s ability to capture the impact of workload characteristics using DNN, RNN, and LLM workloads.
For DNN workloads, we focus on inference tasks for AR applications, specifically image classification, using three widely deployed architectures: MobileNet, Inception, and YOLO, each with distinct compute demands and input sizes: 
\begin{small} 
\begin{center}
\centering
\begin{tabular}{lccc}
\toprule
\textbf{Model} & \textbf{Input Shape} & \textbf{Params (M)} & \textbf{FLOPs (G)} \\
\midrule
MobileNetV2 & $[224,224,3]$ & 3.5 & 0.6 \\
InceptionV4 & $[299,299,3]$ & 42.7 & 6.3 \\
YOLOv8n & $[640,640,3]$ & 3.2 & 8.7 \\ \hline
\end{tabular} 
\end{center}
\end{small} 
All requests are executed using TensorFlow 2.13.0 on frames extracted from a 720p video.
For RNN, we use a Long Short-Term Memory (LSTM) model implemented using TensorFlow’s Keras API for sentiment classification, with the Amazon Alexa Topical Chat dataset \cite{Topical2023} as input. 
For LLM, we use the Llama-3.2-1B model \cite{llama3_2_1B} run via llama.cpp \cite{llama_cpp}, with prompts from the Hugging Face ShareGPT dataset \cite{sharegpt2024} as input.
The specifications of both models are summarized below: 
\begin{small}
\begin{center}
\begin{tabular}{lccc}
\toprule
\textbf{Model} & \textbf{Architecture} & \textbf{Parameters} & \textbf{\# of Layers} \\
\midrule
LSTM & Recurrent & 40 M & 3 \\
Llama-3.2-1B & Transformer & 1.2 B & 16 \\
\bottomrule
\end{tabular}
\end{center}
\end{small}
Since application workloads are often dynamic (e.g., triggered by motion detection), we implement a workload generator that generates requests following a Poisson process, running on a separate machine to avoid performance interference with the device.


\noindent{\textbf{Hardware.}} 
For client devices, we use the NVIDIA Jetson TX2 and Jetson Orin Nano, two commonly deployed platforms in IoT applications, with the former being an earlier model and the latter offering higher performance.
For edge servers, we use a Dell PowerEdge R630 equipped with an NVIDIA A2 GPU and 256GB RAM. 
We also use a server with an NVIDIA RTX 4070 GPU and 32GB RAM. 
Both edge servers are connected to a 1 Gbps network.
These platforms span a wide performance range, from 1.3 TFLOPS (TX2) to 2.6 TFLOPS (Orin Nano), 4.5 TFLOPS (A2), and up to 29.1 TFLOPS (RTX 4070) in peak FP16 throughput. 
All platforms run Ubuntu 18.04. We use the Linux Traffic Control (TC) subsystem to emulate different network bandwidths.
For each accelerator, the effective degree of parallelism $k$ depends on both the number of cores and the workload’s processing and memory demands, and is not directly observable. 
We estimate $k$ by empirically measuring how response time varies with request rate for each workload and identify a value of $k$ that best captures the observed scaling behavior.

\subsection{Parameter Estimation}
\label{sec:para_est}
Our models require several parameters to compute the end-to-end response times. 
We now describe how they are obtained in our evaluation. Additional practical approaches for estimating them are discussed in Sec \ref{sec:service_time}.

\noindent\textbf{Estimating average service times. }
We estimate service times through profiling and a neural network–based prediction approach. 
For DNN workloads, we monitor inference durations using the NVIDIA System Management Interface (nvidia-smi) tool \cite{nvidiasmi}, which reports per-process execution times on the accelerator. 
For RNN and LLM workloads, we profile request latencies on a representative input set separate from the test set and use the average latency as input to the models. 
This profiling data informs our edge-offloading and on-device processing experiments. 
For collaborative processing, a neural network is trained on profiling data to predict service time as discussed in Sec \ref{sec:service_time} to avoid the extensive profiling overhead of all possible split configurations by predicting partial service times.   


\noindent\textbf{Estimating average network delays. }
We estimate transmission delay using the available network bandwidth, measured continuously in real time with iperf \cite{iperfwebsite}. 
In addition to queuing and transmission delays captured by our models, network delays also include propagation and processing delays.
Since we assume that requests are offloaded to a nearby edge server, the propagation delay (caused by signals traveling at the speed of light) is minimal and thus excluded from our models.
Processing delay can be estimated by measuring the round-trip time of small probe packets, and is excluded from our models due to its negligible impact on overall latency.

\noindent\textbf{Estimating system utilization. }
The system utilization $\rho$ is calculated as the ratio of the arrival rate $\lambda$ to the service rate $\mu$ of the system. 
The arrival rate $\lambda$ can be estimated by applying a sliding window over incoming request timestamps. 
The service rate $\mu$ can be estimated from system logs based on the number of completed requests within a given time interval.

\subsection{Impact of Workload Characteristics}
We first study how workload characteristics influence offloading decisions using three DNN models: 1) MobileNetV2 (low compute, small payload); 2) InceptionV4 (high compute, moderate payload); and 3) YOLOv8n (high compute, large payload).
We compare on-device processing on TX2 and Orin Nano against offloading to an A2 GPU using a request rate of 2 requests per second (RPS) and a network speed of 5 Mbps.

\begin{figure}[t]
    \subfloat[\label{plot:tx2_mnet}]{
        \includegraphics[width=0.493\columnwidth]{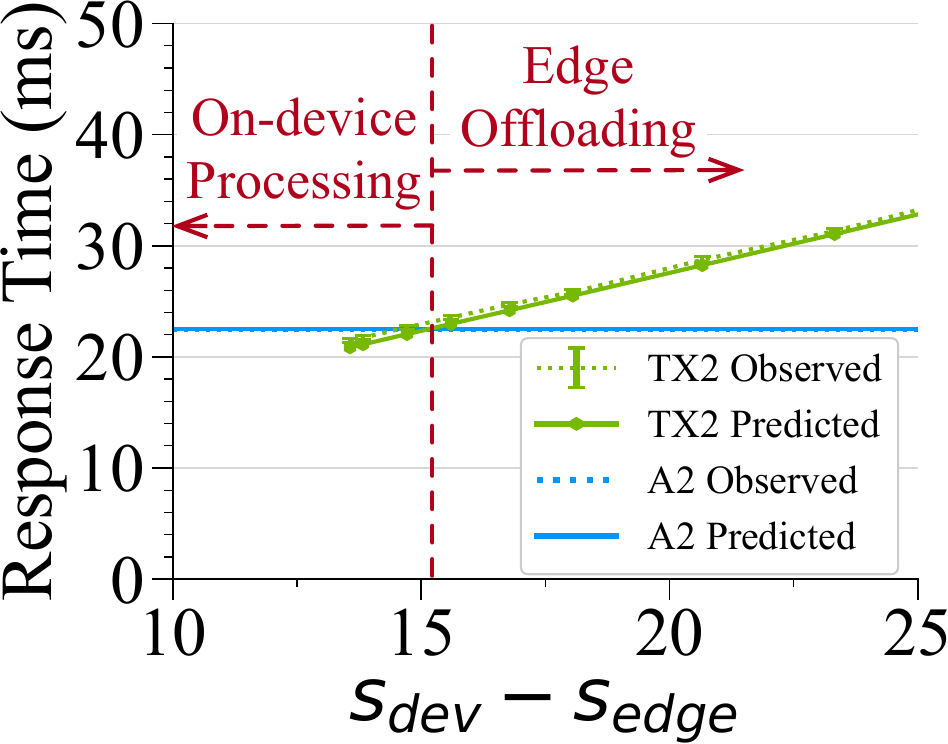}
    }
    \subfloat[\label{plot:orin_mnet}]{
        \includegraphics[width=0.493\columnwidth]{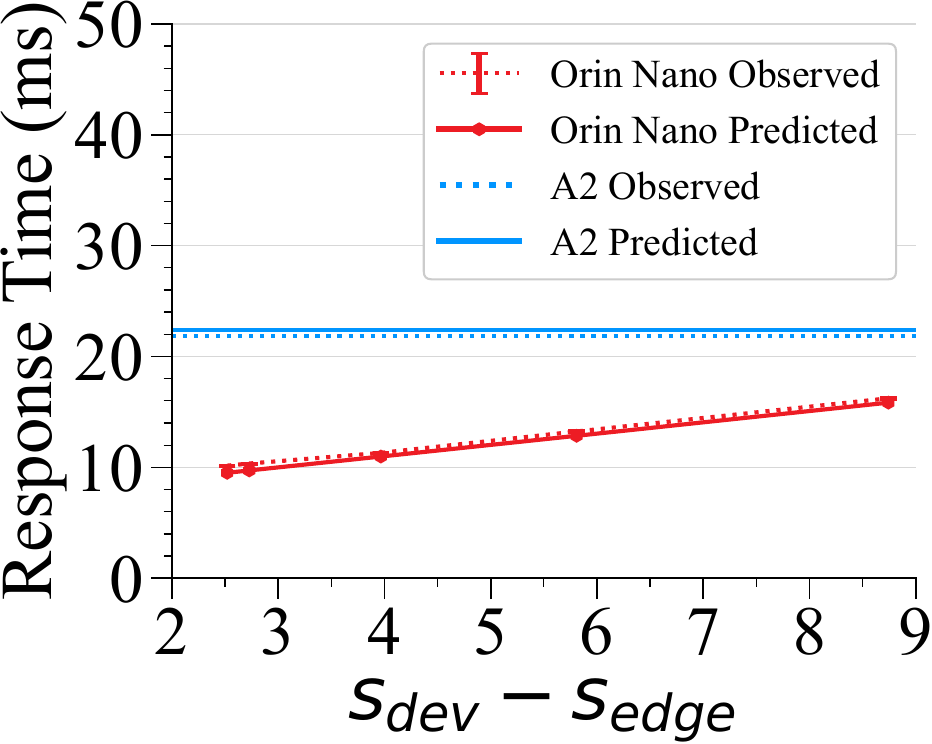}
    }
    \vspace{-.7em}
    \subfloat[\label{plot:tx2_inception}]{
        \includegraphics[width=0.495\columnwidth]{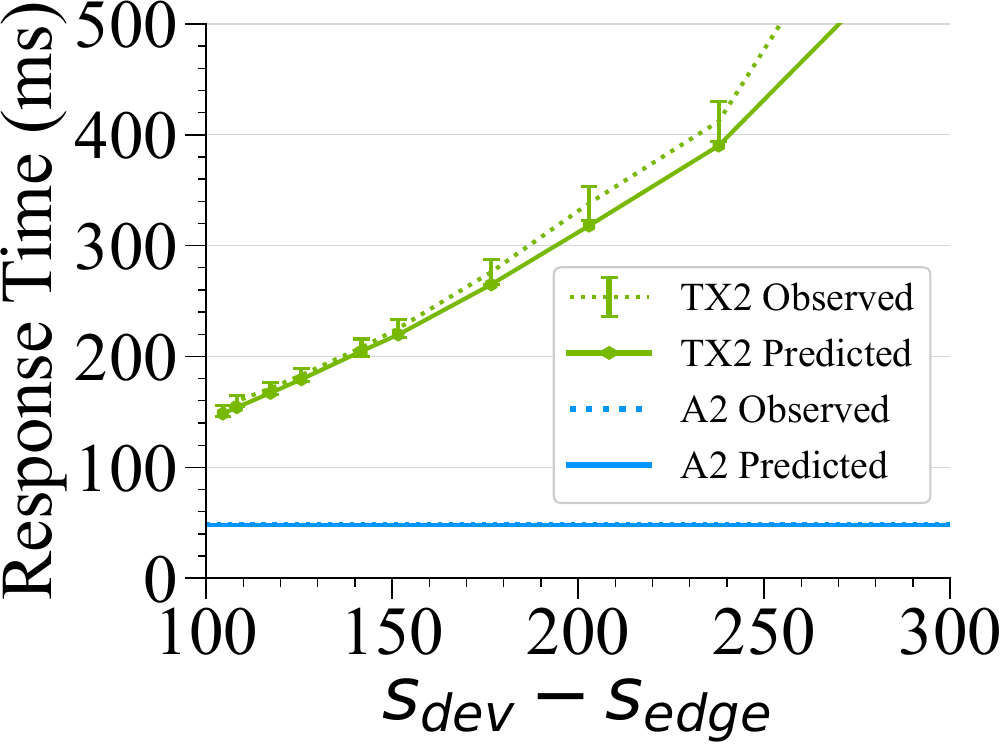}
    }
    \subfloat[\label{plot:orin_inception}]{
        \includegraphics[width=0.495\columnwidth]{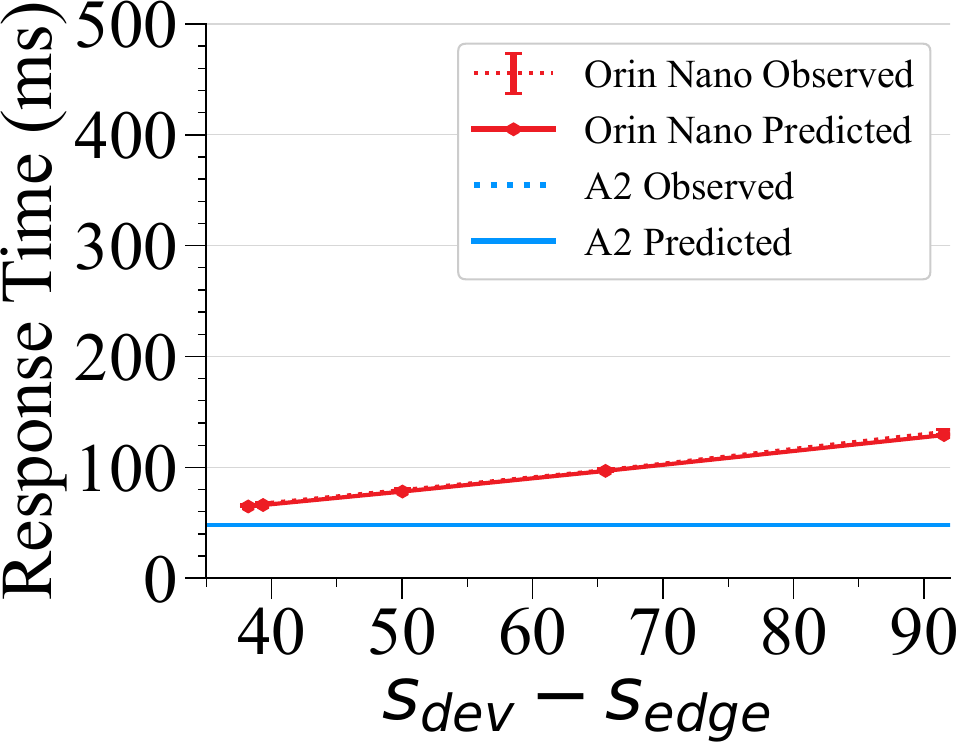}
    }
    \vspace{-.6em}
    \subfloat[\label{plot:tx2_yolo}]{
        \includegraphics[width=0.495\columnwidth]{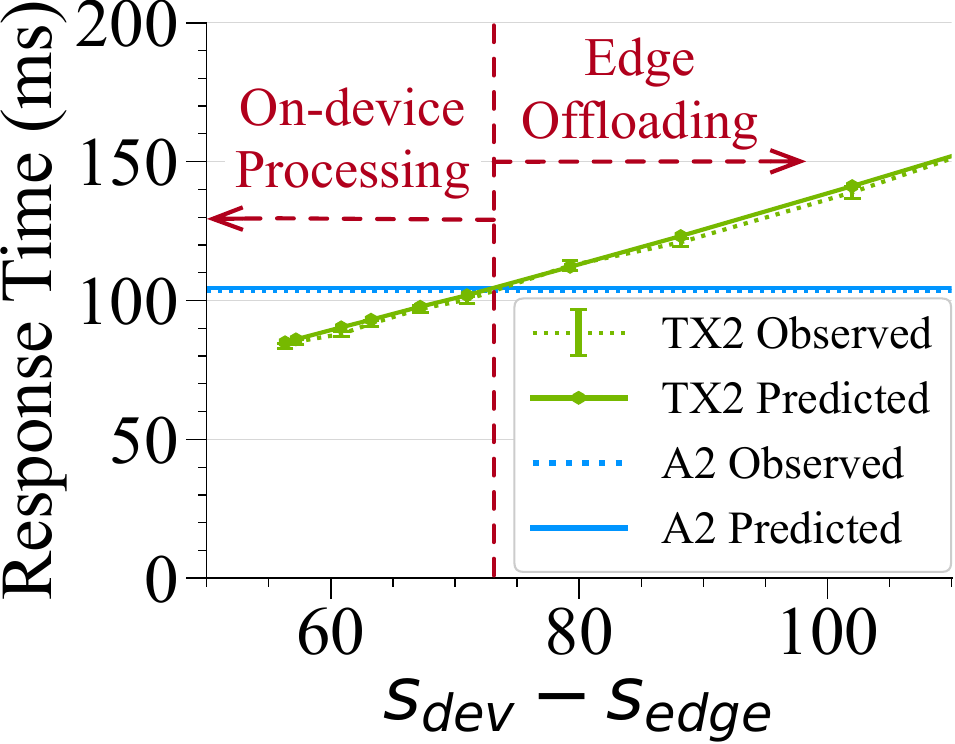}
    }
    \subfloat[\label{plot:orin_yolo}]{
        \includegraphics[width=0.495\columnwidth]{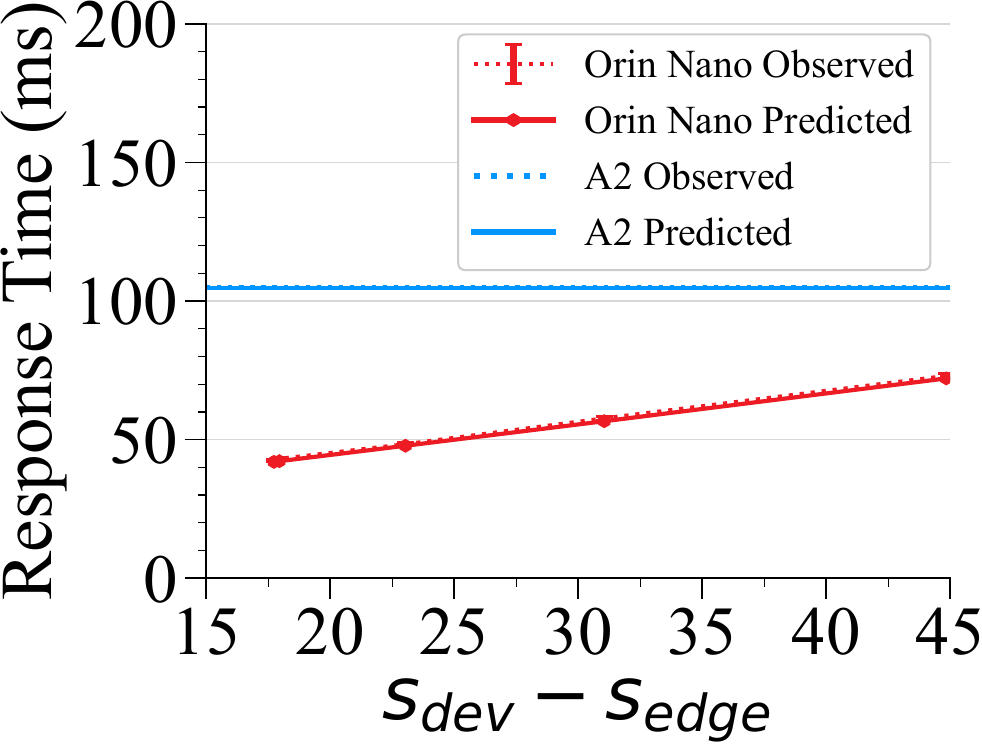}
    }
    \vspace{-.5em}
    \caption{Latency comparison for DNN workloads: (a-b) MobileNetV2, (c-d) InceptionV4, and (e-f) YOLOv8n. }
    \label{plot:varyD}
\end{figure}

Figure \ref{plot:varyD} compares the average latency of on-device processing and edge offloading across workloads and devices.
For each device, we test all available GPU frequency settings to emulate varying device conditions, with each data point representing one setting.
Figures \ref{plot:tx2_mnet} and \ref{plot:orin_mnet} show that for MobileNetV2, TX2 achieves lower average latency in its top three frequency levels, while Orin Nano consistently outperforms offloading across all frequency levels. 
The trend reverses for InceptionV4 as shown in Figures \ref{plot:tx2_inception} and \ref{plot:orin_inception}, where offloading provides lower average latency across all device configurations.
Lastly, Figures \ref{plot:tx2_yolo} and \ref{plot:orin_yolo} show that despite YOLOv8n's high compute demand, TX2 outperforms offloading in its top six frequency levels, while Orin Nano achieves lower average latency across all settings. 
We observe similar results when offloading to an RTX4070 GPU and omit the details here due to space constraints.

Collectively, these results align with the qualitative observations stated in Remark \ref{eq:remark1} and Remark \ref{eq:remark2}. 
While the performance ratio between device and edge is fixed, the absolute advantage of offloading scales with the workload’s computational demands.
For MobileNetV2 (low compute), the edge’s relative speed advantage results in only small absolute gains that can be overshadowed by network delays of offloading.
However, for InceptionV4 (high compute), the same performance ratio yields larger absolute improvements that outweigh the network delays.
Moreover, YOLOv8n (large payloads) favors on-device processing because it avoids any network delay. 
In all cases, the model predictions closely match the observed response times for both approaches, achieving a mean absolute percentage error of 2.2\%, with 91.5\% of predictions falling within $\pm$5\% of the observed latency and 100\% within $\pm$10\%, proving their accuracy.

\noindent\textit{\textbf{Key takeaway.} Offloading provides lower average latency when the edge’s absolute performance gain exceeds the associated network delay. Our models accurately predict the latency for both strategies with a mean absolute percentage error of 2.2\%.}

\subsection{Impact of Complex Deep Learning Models} \label{sec:RNN}
\begin{figure}[t]
    \centering
    \subfloat[\label{plot:LSTM}]{
        \includegraphics[width=.49\columnwidth]{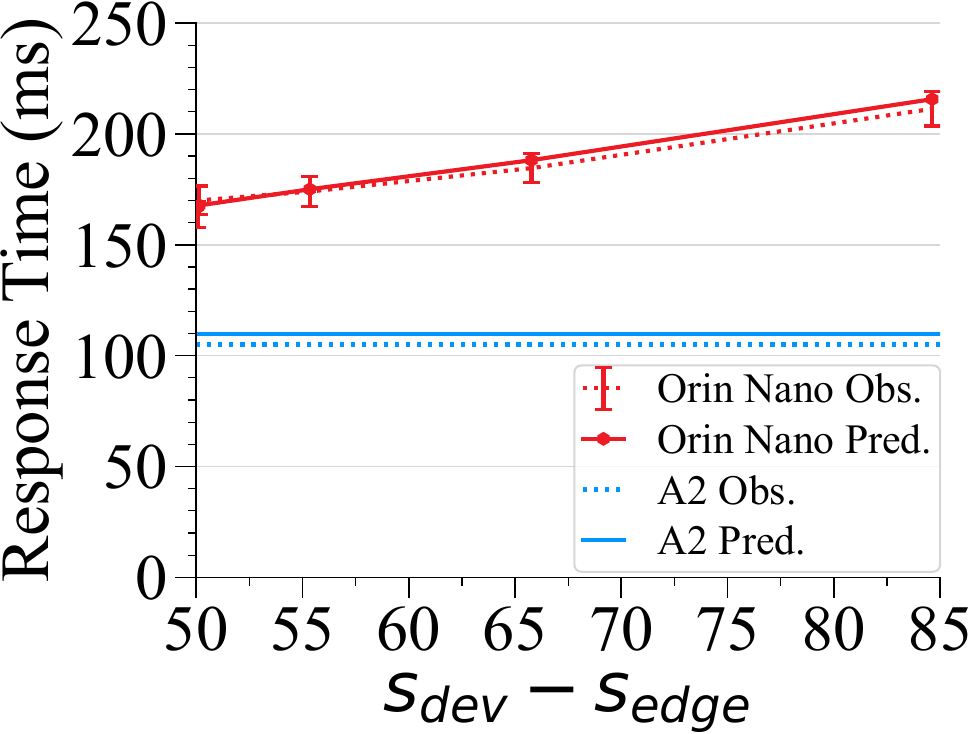}
    }\hfill
    \subfloat[\label{plot:LLM}]{
        \includegraphics[width=.465\columnwidth]{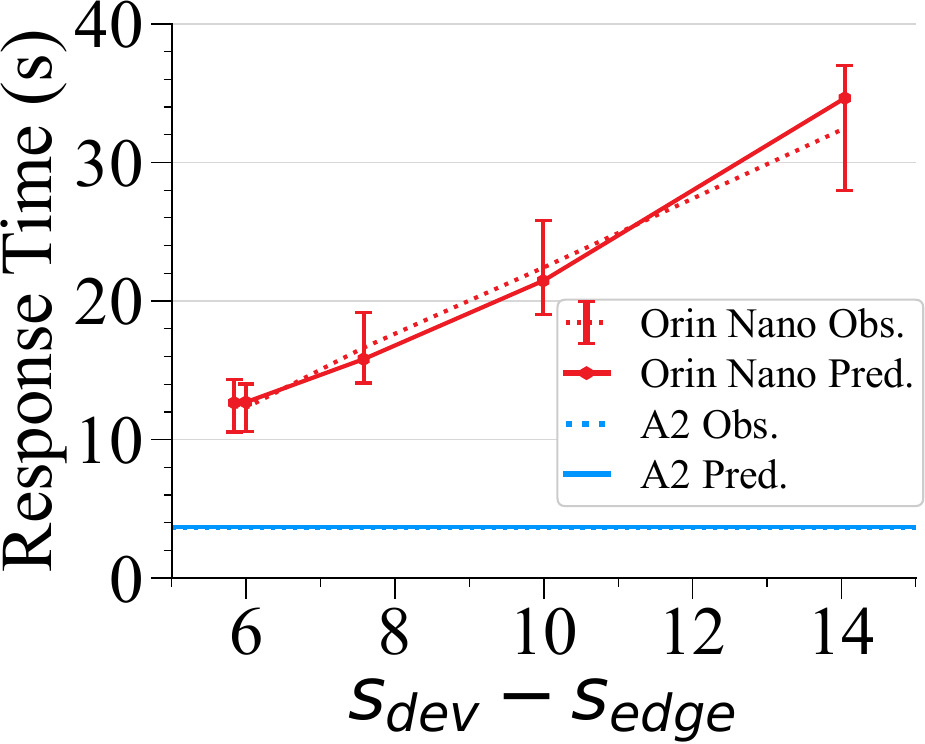}
    }
    \vspace{-1em}
    \caption{Latency comparison of execution strategies for (a) LSTM and (b) Llama-3.2-1B models.}
\label{plots:RNN}
\end{figure}
In addition to DNN workloads, our model accurately captures the latency behavior of other deep learning workloads.
Figure~\ref{plots:RNN} compares the average latency of on-device processing on the Orin Nano against offloading to the A2 GPU for RNN and LLM workloads under moderate loads.
Since service times vary across requests in these models, we use the $M/M/1$ queuing model to account for this variability.
Figure~\ref{plot:LSTM} shows that for the LSTM model, offloading achieves lower average latency as the input utterances are short and their payloads incur only minimal network delay. 
In the LLM case shown in Figure~\ref{plot:LLM}, which has a higher computational demand than LSTM, the latency reduction from offloading becomes even more pronounced, reaching up to 30 seconds.
In both cases, our model predictions closely match the observed latencies.

\noindent\textit{\textbf{Key takeaway.}} Our models generalize to other deep learning workloads by integrating appropriate queuing models.

\subsection{Impact of Network Conditions}
\begin{figure}[t]
    \centering
    \includegraphics[width=\columnwidth]{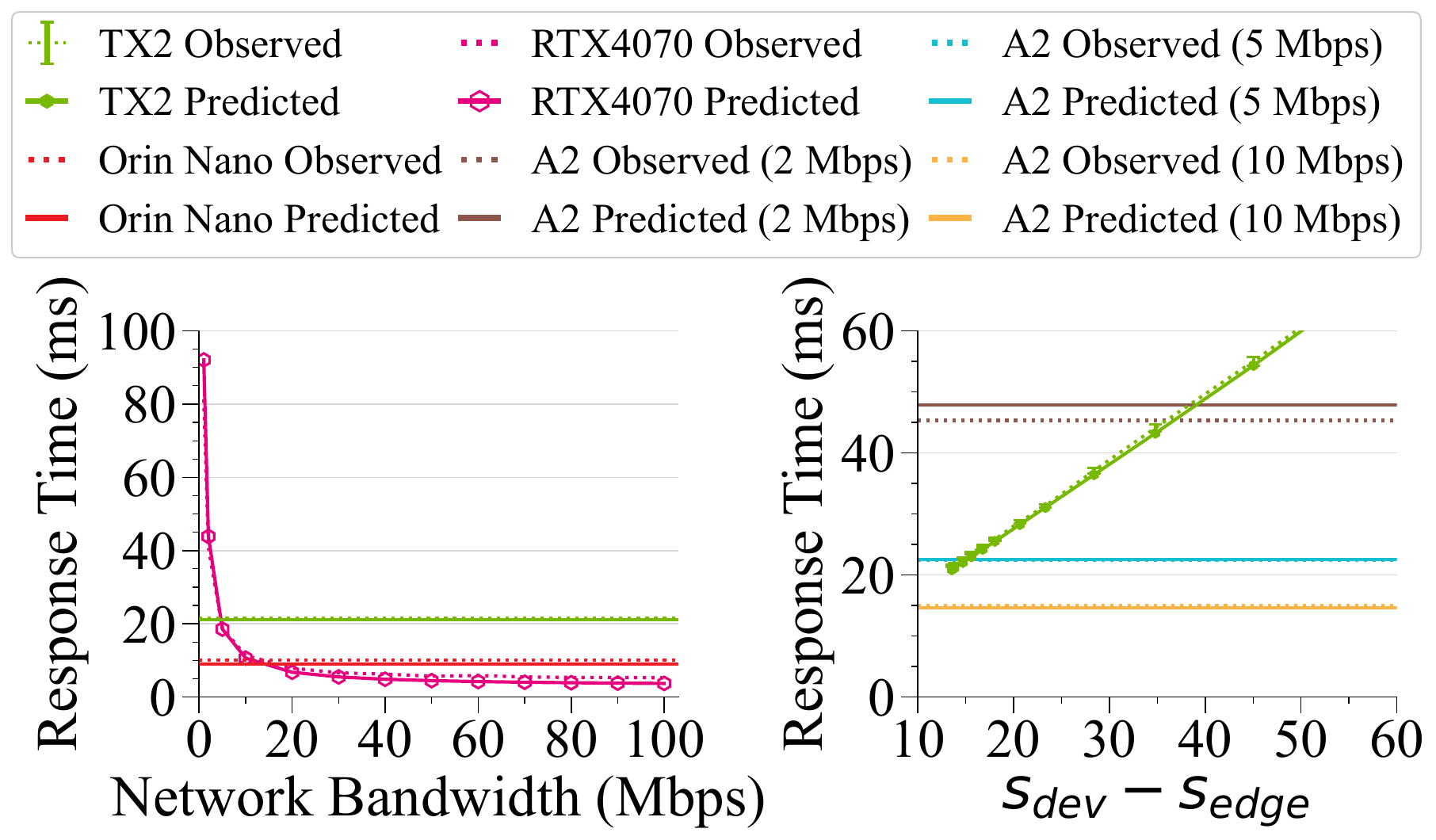} \\
        \vspace{-1em}
    \begin{minipage}{\columnwidth}
        \centering
        \subfloat[\label{plot:mnet_varyNet_4070}]{\hspace{15em}}
        \hfill
        \subfloat[\label{plot:mnet_varyNet_A2}]{\hspace{10em}} 
    \end{minipage}
    \vspace{-1.2em}
    \caption{Latency comparison under varying network bandwidth: (a) RTX4070 as the edge server, (b) A2 as the edge server.}
    \label{plot:mnet_varyNet}
\end{figure}
Figure~\ref{plot:mnet_varyNet_4070} compares on-device processing on TX2 and Orin Nano with offloading to the RTX4070 GPU under varying network bandwidths. 
It shows that the performance crossover point depends on both bandwidth and the performance gap between the device and the edge. 
For the lower-performance TX2, offloading becomes advantageous at 5 Mbps, while the faster Orin Nano requires a higher bandwidth (15 Mbps) for offloading to be beneficial.
Figure \ref{plot:mnet_varyNet_A2} shows that even for a single device, multiple crossover points may exist depending on bandwidth and the device's performance. 
At 5 Mbps, TX2 in its top three frequency levels outperforms offloading, but as its performance drops beyond those levels, offloading becomes advantageous.
In all cases, our models closely predict the crossover points, showing their accuracy in capturing the tradeoff between the computation speedup from edge offloading and network delay.

\noindent\textit{\textbf{Key takeaway.} At high network bandwidths, edge offloading becomes favorable as its computational advantage outweighs the smaller network delays.}

\subsection{Impact of Collaborative Processing}
\begin{figure*}[t]
\centering
\subfloat[\label{plot:split}]{
    \includegraphics[width=0.5248\columnwidth]{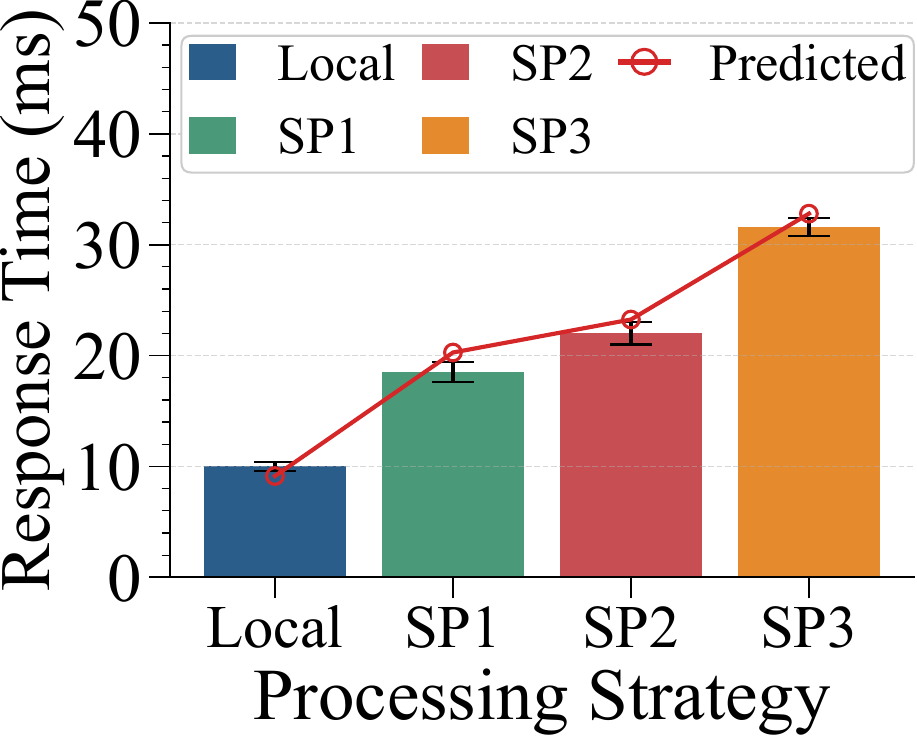}
}
\hspace{.8em}
\subfloat[\label{plot:mnet_varyRPS}]{
    \includegraphics[width=0.8922\columnwidth]{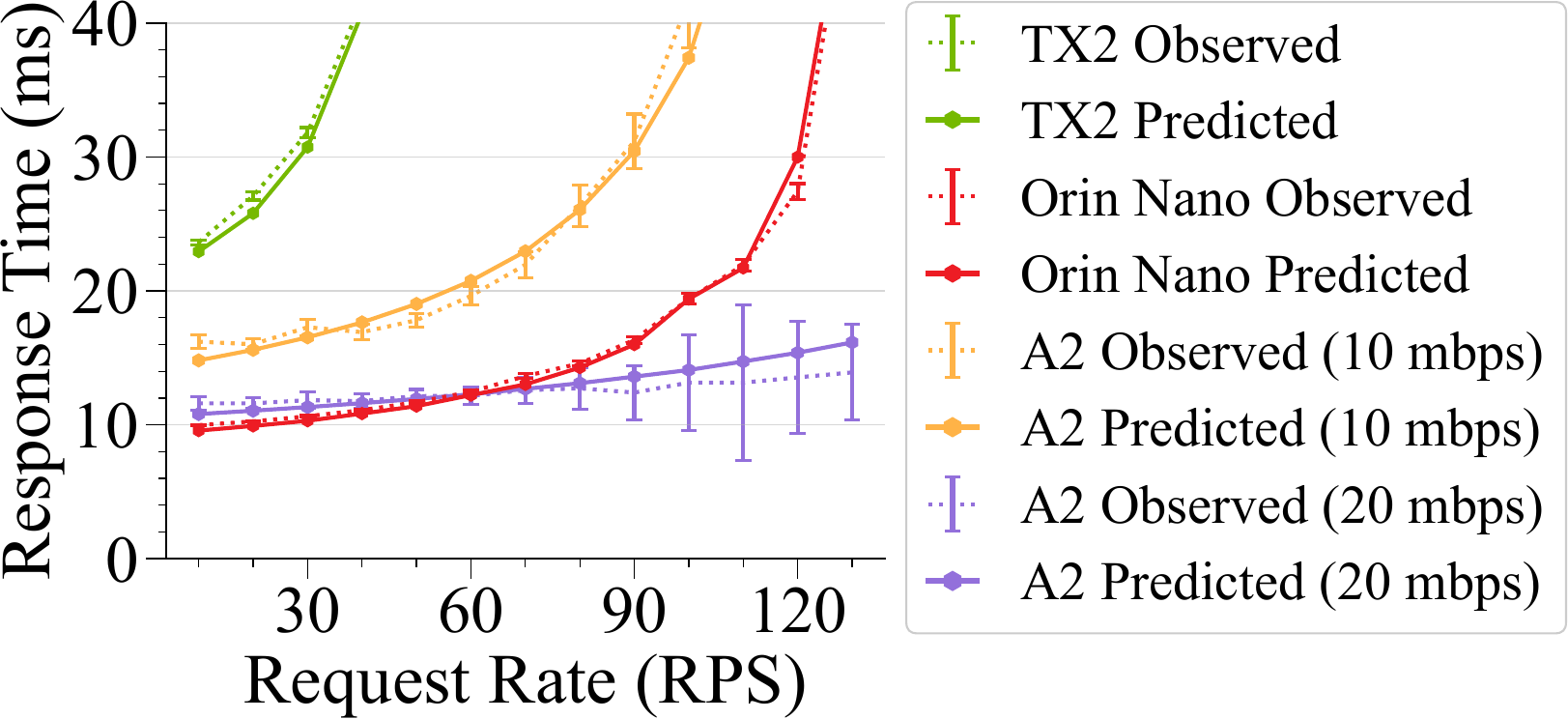}
}
\hspace{.8em}
\subfloat[\label{plot:inception_multi_tenant}]{
    \includegraphics[width=0.5402\columnwidth]{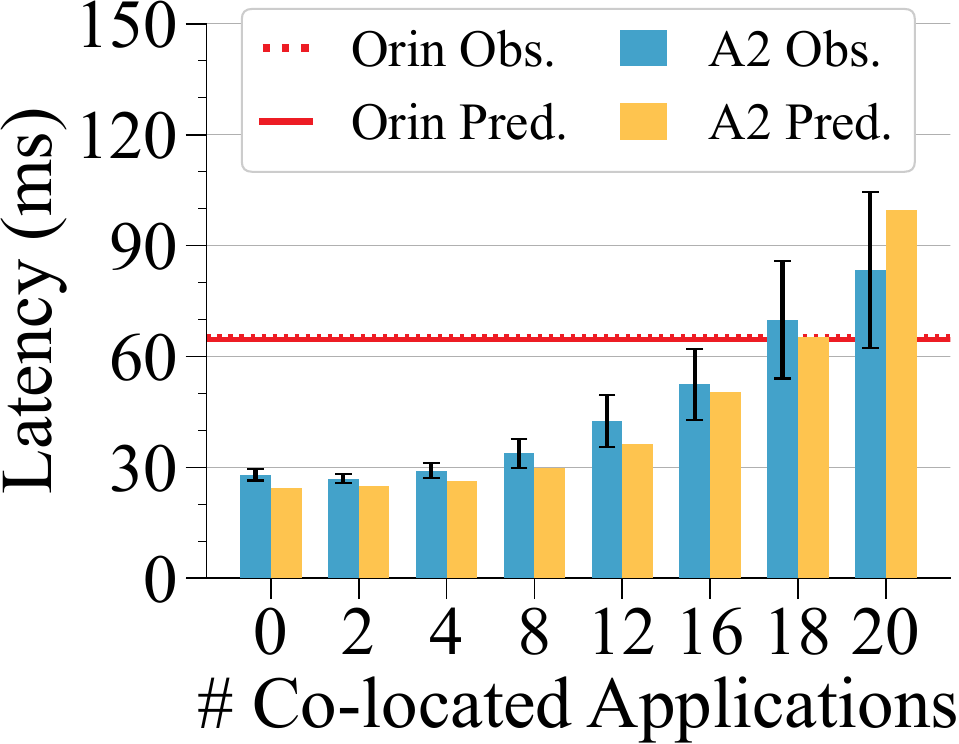}
}
\vspace{-.75em}
\caption{Latency comparisons of execution strategies under different: (a) collaborative processing configurations, (b) request rates, and (c) numbers of co-located applications.}
\label{fig:merge}
\end{figure*}
To validate our models' ability to capture collaborative (i.e., split) processing behavior, we split the execution of MobilNetV2 across Orin Nano and A2 over a 50 Mbps network.
We evaluate three split points (SPs) where the model is progressively offloaded to the edge and compare them to full on-device execution. 
For each split point, the early neural network layers are executed on the device, and the intermediate output and remaining computation are offloaded to the edge.
To compute the end-to-end latency, we use the combined analytic model discussed in Sec \ref{sec:discussions}.
Figure~\ref{plot:split} plots the observed latencies of each configuration alongside model predictions. 
As shown, the latency increases with the amount of compute offloaded.
This is because the intermediate results of later layers grow in size, adding to the network transfer overhead. 
As seen, for MobileNetV2 and the chosen network configuration, local processing is faster than split processing. 
Note that other models will exhibit other types of behaviors with collaborative processing, which we omit due to space constraints.  
Across all split points, the model-predicted latencies closely align with observed values. 

\noindent\textit{\textbf{Key takeaway.} 
Our models can be extended to accurately predict the performance of device-edge collaborative processing. 
}

\subsection{Impact of Request Rate}
Figure \ref{plot:mnet_varyRPS} shows how the request rate $\lambda$ affects offloading decisions for the MobileNetV2 workload under 10 Mbps and 20 Mbps networks, illustrating two distinct behaviors depending on the network bandwidth.
At 10 Mbps, Orin Nano outperforms offloading across all RPS. 
However, at 20 Mbps, a performance crossover occurs: 
on-device processing is faster at low load (RPS $\leq$ 60), while offloading becomes more efficient beyond that point. 
This result validates the competing scaling effects of $\lambda$ discussed in Lemma \ref{eq:d_vs_e_lemma}. 
Although both the network queuing delays (for offloading) and processing queuing delays grow with $\lambda$, the rate at which the delays increase differs depending on the hardware, workload, and network configurations. 
At 10 Mbps, the network queuing delay grows faster than the device queuing delay as $\lambda$ increases, favoring on-device processing. 
Conversely, at 20 Mbps, device queuing delays become dominant at higher $\lambda$, shifting the advantage toward offloading. 
Overall, the predicted latencies of our models closely match the observed latencies within confidence intervals, demonstrating their accuracy.

\noindent\textit{\textbf{Key takeaway.} As load increases, on-device processing becomes more advantageous under slower networks, while offloading provides lower average latency on faster networks. }

\subsection{Impact of Multi-Tenancy}

Figure~\ref{plot:inception_multi_tenant} compares the average latency of the two strategies as the number of co-located InceptionV4 applications on the edge increases, where each application receives 2 RPS from its corresponding client. 
The result shows that when only a single application is deployed on the edge, offloading yields lower average latency than on-device processing.
However, as the number of co-located applications increases, offloading latency increases as well. 
This increase is caused by resource contention among applications because the A2 GPU lacks isolation mechanisms. 
As a result, there is a crossover point at 18 applications, beyond which on-device processing becomes more efficient.
Note that the on-device processing latency remains stable regardless of the number of concurrent applications, as each device operates independently.
In all cases, our models closely predict the latencies for both strategies.

\noindent\textit{\textbf{Key takeaway.} 
Offloading average latency increases with the number of co-located applications due to resource contention on edge servers lacking GPU isolation.}


\section{Models in Action} 
In this section, we demonstrate how a resource manager can leverage our models to adaptively switch between offloading and on-device processing. 
We outline its algorithm and demonstrate its effectiveness through two case studies that showcase how the models enable adaptation to network variability and to multi-tenant edge servers.

\subsection{Model-Driven Adaptive Resource Management}
\begin{small}
\begin{algorithm}[!t]
\caption{Model-Driven Adaptive Offloading}
\label{alg:selection}
\KwIn{
\\ $D_{req}, D_{res}$: Request/result payload size \vspace{0.1em}
\\ $s_{dev}^{proc}, s_{edge}^{proc}$: Service times of device and edge
\\ $\lambda_{dev}$: Device task arrival rate \vspace{0.1em}
\\ $B$: Network bandwidth of device \vspace{0.1em}
\\ $\mathcal{E}$: Number of edge servers \vspace{0.1em}
\\ $\{\mu_{edge,E}^{proc}\}_{E=1}^{\mathcal{E}}$: Edge server aggregated service rates 
\\ $\{\lambda_{edge,E}\}_{E=1}^{\mathcal{E}}$: Edge server aggregated arrival rates \vspace{0.2em}
}
\KwOut{
\\ Execution strategy (on-device or edge server $E^*$)}
\BlankLine
\tcp{Predict on-device processing latency}
$\mu_{dev}^{proc} \gets 1/s_{dev}^{proc} $\\
$T_{dev} \gets M/D/1(\lambda^{proc}_{dev},\mu_{dev}^{proc}) + s_{dev}^{proc}$

\tcp{Predict edge offloading latency}
$T_{net}^{req} \gets M/M/1(\lambda_{dev},B/D_{req}) + D_{req}/B$\\
$T_{net}^{res} \gets M/M/1(\lambda_{edge,E},B/D_{res}) + D_{res}/B$\\
\For{$E \gets 1$ \KwTo $\mathcal{E}$}{
    $T_{edge,E} \gets  T_{net}^{req} + M/G/1(\lambda_{edge,E},\mu_{edge,E}^{proc}) + s_{edge}^{proc} + T_{net}^{res}$
}
\tcp{Select optimal strategy} 
\eIf{$T_{dev} < \min(\{T_{edge,E}\}_{E=1}^{\mathcal{E}})$}{
  \Return{on\_device\_processing()}
}{
    $E^* \gets \operatorname*{arg\,min}_E T_{edge,E}$ \\
  \Return{offload($E^*$)} \tcp{Offload to $E^*$}
}
\end{algorithm}
\end{small}

The resource manager runs locally on the device and periodically collects runtime metrics including network bandwidth, edge server load, and request arrival rate as discussed in Sec \ref{sec:para_est}. 
At each epoch, it inputs these metrics into our models to estimate the average latency for each strategy, and selects the strategy with the lowest predicted latency to execute requests.
Algorithm~\ref{alg:selection} outlines the decision-making process of the resource manager.
Lines 1-2 predict the average latency of on-device processing based on the current arrival rate and the device's estimated service time.
Lines 3-6 compute the predicted offloading latency for each edge server, considering server load, service rates, and network conditions.
Finally, Lines 7-11 execute the request using the strategy with the lower predicted latency.

\subsection{Case 1: Fluctuating Network Conditions}
\begin{figure}[!t]
    \centering
    \includegraphics[width=.99\columnwidth]{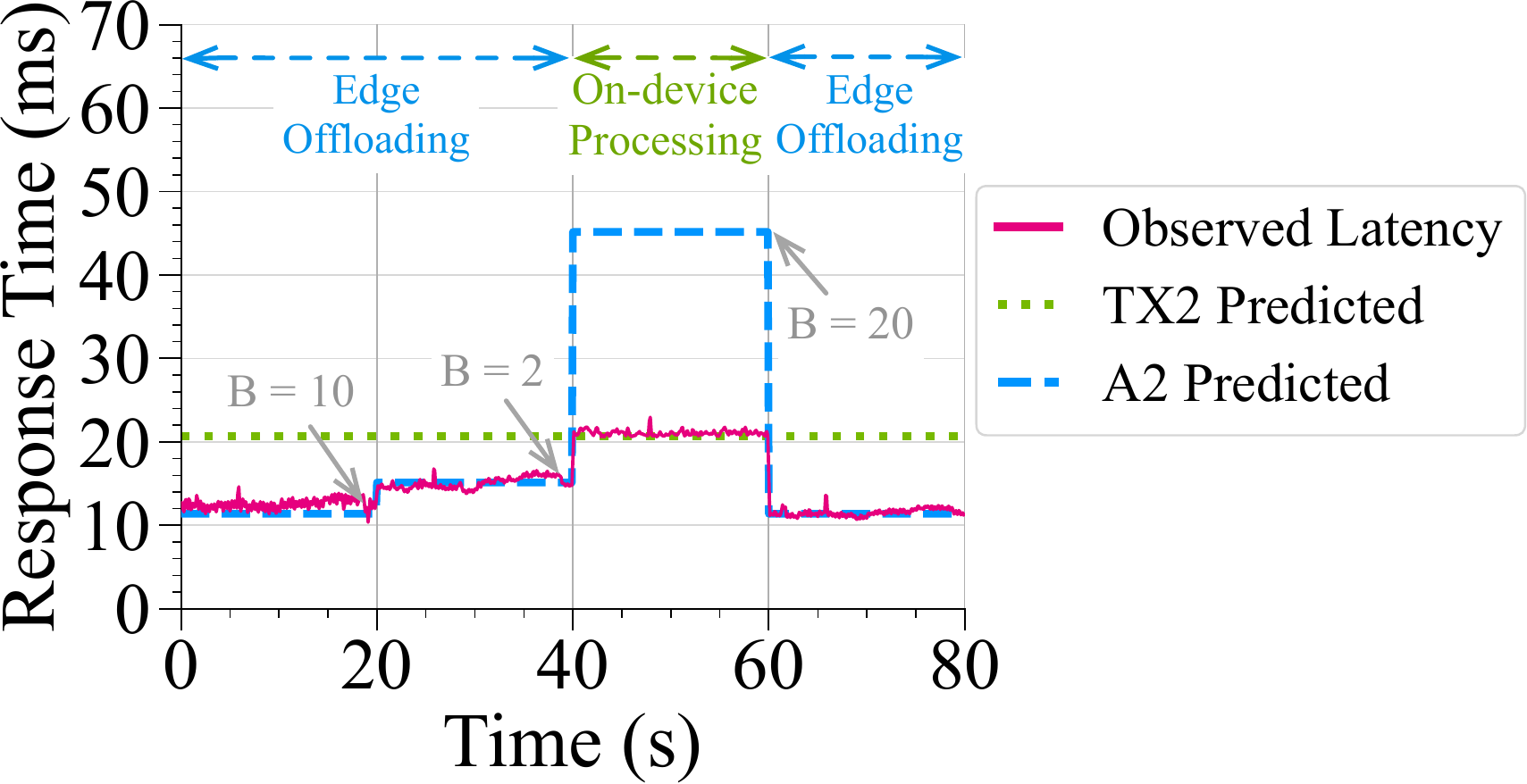}
    \vspace{-1em}
    \caption{Latencies under varying network bandwidth as the resource manager selects between the two strategies. }
    \label{plot:network_dynamics}
\end{figure}

We first demonstrate how the resource manager can leverage our models to adapt to dynamic network conditions, using TX2 as the device, A2 as the edge, with MobileNetV2 requests arriving at a fixed rate of 10 RPS.
Figure~\ref{plot:network_dynamics} shows the observed latency using the execution strategies chosen by the resource manager, alongside the model-predicted latency of each strategy.
Initially, the device is configured with a 20 Mbps network to represent a stable 5G connection. 
Since the predicted offloading latency is lower than on-device processing, the manager chooses to offload requests. 
To emulate real-world network variation, we reduce the bandwidth to 10 Mbps at t = 20 to reflect moderate signal degradation. 
Here, despite the reduced bandwidth from 20 to 10 Mbps, the predicted latency for offloading remains lower, and the execution strategy remains unchanged. 
At t = 40, we further reduce the bandwidth to 2 Mbps to emulate poor coverage. 
Since the predicted offloading latency (45 ms) now exceeds the predicted on-device processing latency (21 ms), the resource manager switches to on-device processing.
Finally, we restore the bandwidth to 20 Mbps at t = 60 to represent signal recovery, and the resource manager switches back to offloading.
Throughout this sequence, our models closely predict the observed latencies, validating their accuracy. 

\noindent\textit{\textbf{Key takeaway. }Our models enable dynamic adaptation to changing network conditions by switching between execution strategies.}

\subsection{Case 2: Multi-tenant Edge Servers}
\begin{figure}[!t]
    \centering
    \includegraphics[width=.99\columnwidth]{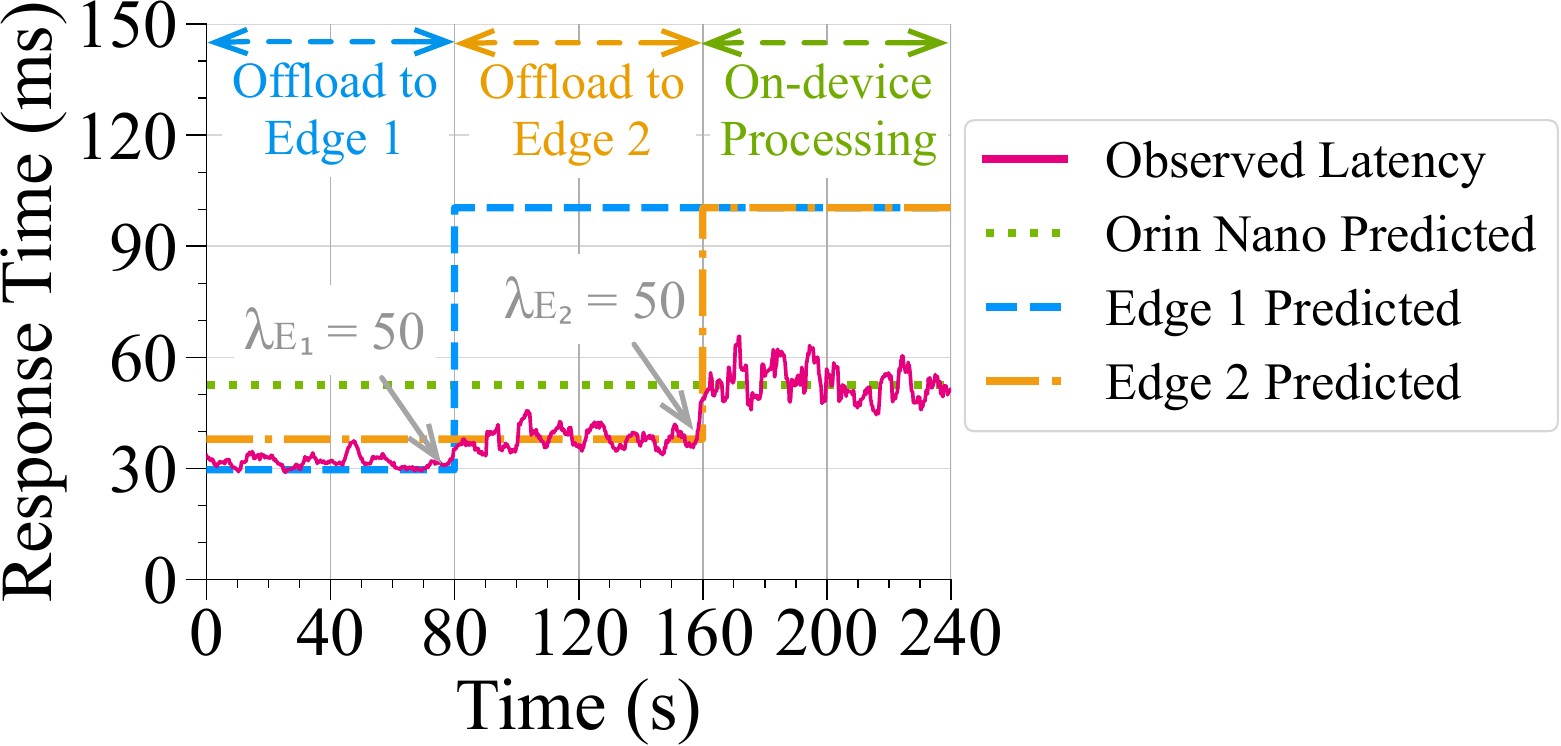}
    \vspace{-1em}
    \caption{Latencies under varying edge server workloads as the resource manager adapts. }
    \label{plot:mutli_tenant_rps}
\end{figure}
We now showcase how the resource manager can leverage our models to adapt to workload changes across multi-tenant edge servers. 
We use two edge servers, $E_1$ and $E_2$, each serving YOLOv8n inference requests using an A2 GPU.
Initially, the server request rates are set to $\lambda_{E_1} = 10$ and $\lambda_{E_2} = 30$, while the TX2 running the resource manager has a request rate of $\lambda_{dev} = 10$. 
Figure~\ref{plot:mutli_tenant_rps} shows the observed latencies alongside the predicted average latencies for offloading to $E_1$, $E_2$, and on-device processing.
At t = 0, since the predicted latency for offloading to $E_1$ is the lowest, the resource manager offloads requests to $E_1$. 
At t = 80, we add an additional device offloading request to $E_1$ to increase $\lambda_{E_1}$ to 50 to emulate workload fluctuations. 
In response, the resource manager adapts by offloading to the less loaded server $E_2$.
Similarly, at t = 160, we increase $\lambda_{E_2}$ to 50. 
Now, as the predicted latency of on-device processing becomes the lowest when both servers have a request rate of 50 RPS, the resource manager shifts to on-device processing.
This sequence highlights how our model enables the resource manager to dynamically adapt to fluctuations in edge server workloads.

\noindent\textit{\textbf{Key takeaway. } Our models enable the resource manager to dynamically adapt to workload changes in multi-tenant edge servers to maintain low latency.}
\section{Related Work}
\noindent\textbf{Edge offloading. }
Edge offloading is widely adopted in many applications to enable resource-constrained devices to perform compute-intensive tasks by leveraging edge resources \cite{Davis2004,Jang2018,Lai2017,Li2024,Ogden2021,Ren2019,Mei2012,Liu2018,zhou2017,Wang2018,Zhang2024MILCOM,Tkachenko2023,Ogden2023,Zilic2025}. 
For example, Furion \cite{Lai2017} selectively offloads compute-intensive components of mobile AR pipelines to the edge. 
Li et al. \cite{Li2024} introduce an evolutionary algorithm to dynamically determine the optimal offloading location of DNN layers under energy constraints.
In contrast, our work uses a model-driven approach to provide end-to-end latency predictions and quantitative performance crossover predictions for accelerator-driven workloads and generalizes to device-edge collaborative processing. 

\noindent\textbf{Model-aware resource allocation. }
Queuing models have been widely used to predict the performance of traditional CPU-bound workloads \cite{Jiang2021, WangB2024, Tantawi1985, Ambati2020, Gandhi2019, Shen2011, Urgaonkar2005}. 
More recently, there has been growing interest in modeling the performance of hardware accelerators. 
For example, Merck et al. \cite{Merck2019} model GPU batching latency in DNN inference workloads. Liang et al. \cite{Liang2023} develop performance models for inference on accelerators to guide DNN placement via an online knapsack algorithm.
In contrast, our models apply not only to inference workloads but also to other accelerator-driven workloads as shown in Sec \ref{sec:RNN}. 

\noindent\textbf{Modeling distributed workloads. }
Several studies have developed analytic models for distributed device–edge–cloud environments. 
FogQN \cite{Tadakamalla2018} uses multi-class queuing networks to model fog and cloud computing to determine the optimal fraction of data processing executed at the cloud versus at fog servers. 
Loghin et al. \cite{Loghin2019} model MapReduce applications in hybrid edge–cloud environments.
However, they abstract away accelerator-specific queuing behavior and do not account for multi-tenancy at the edge.
Our work complements existing efforts by addressing these missing dimensions through a unified queuing-based modeling framework that informs resource management decisions.




\section{Conclusions}
\label{sec:conclusion}
This paper presented analytic models for on-device processing and edge offloading with accelerators, providing performance predictions for adaptive execution decisions. 
We validated their accuracy across a range of scenarios and developed a resource manager that applies these predictions to adapt to real-world dynamics.
As future work, we plan to extend the models to support multi-stage pipelines represented as DAG structures, capturing inter-stage dependencies.

\bibliographystyle{ACM-Reference-Format}
\bibliography{main}

\end{document}